\documentclass[technote]{IEEEtran}

\usepackage{amsmath,graphicx}
\newtheorem{theorem}{\bf Theorem}[section]
\newtheorem{proposition}{\bf Proposition}[section]
\newtheorem{lemma}{\bf Lemma}[section]
\newtheorem{corollary}{\bf Corollary}[section]

\begin{document}
\title{A Bit of Information Theory, and the Data Augmentation Algorithm Converges}

\author{Yaming~Yu, {\it Member, IEEE}
\thanks{Yaming Yu is with the Department of Statistics, University of California, Irvine, CA, 92697-1250, USA
(e-mail: yamingy@uci.edu).  This work is supported in part by a start-up fund from the Bren School of Information and 
Computer Sciences at UC Irvine.}
}

\maketitle

\begin{abstract}
The data augmentation (DA) algorithm is a simple and powerful tool in statistical computing.  In this note basic 
information theory is used to prove a nontrivial convergence theorem for the DA algorithm. 
\end{abstract}


\section{Background}
In many statistical problems we would like to sample from a probability density $\pi(x, y)$, e.g., the joint 
posterior of all parameters and latent variables in a Bayesian model.  When $\pi(x, y)$ is complicated, direct 
simulation may be impractical; however, if the conditional densities $\pi_{X|Y}(x|y)$ and $\pi_{Y|X}(y|x)$ are 
tractable, the following algorithm is an intuitively appealing alternative.  Draw $(X, Y)$ from an initial density 
$p^{(0)}(x, y)$, and then alternatingly replace $X$ by a conditional 
draw given $Y$ according to $\pi_{X|Y}(x|y)$, and $Y$ by a conditional draw given $X$ according to $\pi_{Y|X}(y|x)$; this is a 
crude description of the {\it data augmentation} (DA) algorithm of Tanner and Wong \cite{TW} (see also \cite{MV99}, 
\cite{VM01} and \cite{YM}), a powerful and widely used method in statistical computing.

It is not immediately obvious that iterates of the DA algorithm should approach the target $\pi(x, y)$.  To show convergence, 
one usually appeals to Markov chain theory (Tierney \cite{Tier}), which says that (roughly) if a Markov chain is irreducible 
and 
aperiodic, and possesses a stationary distribution, then it converges to that distribution.  Such results are often stated in 
terms of the {\it total variation distance}, defined for two densities $p$ and $q$ as 
$$V(p, q)=\int |p-q|.$$
Because iterates of the DA algorithm form a Markov chain, they converge in total variation under some regularity conditions.

Total variation, of course, is not the only discrepancy measure.  There is actually another discrepancy measure that is 
natural for the problem, yet rarely explored.  Recall that the 
{\it relative entropy}, or Kullback-Leibler divergence, between two densities $p$ and $q$ is defined as
$$D(p|q)=\int p\log(p/q).$$
It is related to $V(p, q)$ via the well-known Pinsker's inequality
$$D(p|q)\geq \frac{1}{2}V^2(p, q),$$
so that for a sequence of densities $p_t,\ t=0, 1, \cdots$, $\lim_{t\rightarrow\infty} D\left(p_t|p_\infty\right)=0$ implies 
$\lim_{t\rightarrow\infty} V\left(p_t, p_\infty\right)=0$.  Other useful properties of relative entropy can be found in Cover 
and Thomas \cite{Cover}.

It is the purpose of this note to analyze the DA algorithm in terms of relative entropy and present a short proof of a 
convergence result (Theorem \ref{main}) using simple information theoretic techniques.

\section{Main Result}
Let $\mu\times \nu$ be a product measure on a joint measurable space $(\mathcal{X}\times \mathcal{Y}, \mathcal{F}\times 
\mathcal{G})$.  Suppose the target density $\pi(x, y)$ with respect to $\mu\times \nu$ satisfies $\pi(x, y)>0$ for all $(x, 
y)\in \mathcal{X}\times \mathcal{Y}$ (in statistical applications often $\mathcal{X}$ and $\mathcal{Y}$ are subsets of 
Euclidean spaces and each of $\mu$ and $\nu$ is either Lebesgue measure or the counting measure).  Formally, 
given an initial density $p^{(0)}(x, y)$, the DA algorithm generates a sequence of densities $p^{(t)}(x, y),\ t\geq 
0$, where ($p^{(t)}_X(x)=\int_{\mathcal{Y}} p^{(t)}(x, y)\, {\rm d}\nu(y)$, for example)
\begin{equation}
\label{def1}
p^{(t+1)}(x, y)=\left\{\begin{array}{ll} p^{(t)}_X(x)\pi_{Y|X}(y|x),& t\ {\rm odd};\\
p^{(t)}_Y(y)\pi_{X|Y}(x|y),& t\ {\rm even}.\end{array}\right.
\end{equation}

\begin{theorem}
\label{main}
If $\pi(x, y)>0$ for all $(x, y)\in \mathcal{X}\times \mathcal{Y}$, and $D\left(p^{(0)}|\pi\right)<\infty$, then iterates of 
the DA algorithm (\ref{def1}) converge in relative entropy, i.e.,
$$\lim_{t\rightarrow\infty}D\left(p^{(t)}|\pi\right)=0,$$
and $\lim_{t\rightarrow\infty}V\left(p^{(t)},\pi\right)=0$ necessarily.
\end{theorem}

The condition $\pi(x, y)>0,\ (x, y)\in \mathcal{X}\times \mathcal{Y},$ can be weakened, and the result can be 
generalized to the Gibbs sampler (\cite{GG} \cite{GS}); see Yu \cite{Y}.  Note that the conditions of Theorem 
\ref{main} are 
already weaker than those of Schervish and Carlin \cite{SC}, for example (see also Liu et al.\ \cite{LWK95}), although Theorem 
\ref{main} does not give a qualitative rate of convergence.  As a general comment, the 
approach taken here complements the more traditional $L_2$ approach (Amit \cite{A}) that studies the Gibbs sampler in the 
Hilbert space of square integrable functions. 

Section III provides a short, self-contained proof of Theorem \ref{main}. The main tools (Lemmas 
\ref{tri}--\ref{tri2}) exploit the information geometry of the DA algorithm.  Although relative entropy does not define a 
metric, 
it behaves like squared Euclidean distance.  See Csisz\'{a}r \cite{C75}, Csisz\'{a}r and Shields \cite{CS}, and Csisz\'{a}r 
and Mat\'{u}s \cite{C03} for the notions of 
{\it I-projection} and {\it reverse I-projection} that explore such properties in broader contexts. 

\section{Proof of Theorem \ref{main}}
In this section let $p^{(t)}$ be a sequence of densities generated according to (\ref{def1}) with 
$D\left(p^{(0)}|\pi\right)<\infty$.  Lemma \ref{tri} captures the intuition that each iteration is a projection (more 
precisely, a reverse I-projection) onto the set of densities with a given conditional.  The proof is simple and hence 
omitted.

\begin{lemma}
\label{tri}
For all $t\geq 0$,
$$D\left(p^{(t)}|\pi\right)=D\left(p^{(t)}|p^{(t+1)}\right)+D\left(p^{(t+1)}|\pi\right).$$
\end{lemma}

According to Lemma \ref{tri}, $D\left(p^{(t)}|\pi\right)$ can only decrease in $t$ (this holds for Markov chains in general).  
However, it does not imply $D\left(p^{(t)}|\pi\right)\downarrow 0$.  To prove the theorem we need further analysis.

\begin{lemma}
Let $t\geq 1$ and $n\geq 1$.  If $n$ is even then
\begin{equation}
\label{ineq0}
D\left(p^{(t)}|p^{(t+n)}\right)\leq D\left(p^{(t)}|p^{(t+n-1)}\right);
\end{equation}
if $n$ is odd then
\label{prop}
\begin{equation}
\label{tri1}
D\left(p^{(t)}|p^{(t+n)}\right)=D\left(p^{(t)}|p^{(t+1)}\right)+D\left(p^{(t+1)}|p^{(t+n)}\right).
\end{equation}
\end{lemma}

\begin{proof} To prove (\ref{ineq0}), without loss of generality assume $t$ is odd.  Since $n$ is even, $p^{(t)}$ and 
$p^{(t+n)}$ have the same conditional 
$p^{(t)}_{X|Y}=p^{(t+n)}_{X|Y}=\pi_{X|Y}$, whereas $p_Y^{(t+n)}=p_Y^{(t+n-1)}$ by (\ref{def1}).  We have
\begin{align*}
D\left(p^{(t)}|p^{(t+n)}\right)&=D\left(p_Y^{(t)}|p_Y^{(t+n)}\right)\\
&=D\left(p_Y^{(t)}|p_Y^{(t+n-1)}\right)\\
&\leq D\left(p^{(t)}|p^{(t+n-1)}\right),
\end{align*}
the last inequality being a basic property of relative entropy (Cover and Thomas \cite{Cover}).  The proof of (\ref{tri1}), 
omitted, 
is the same as that of Lemma \ref{tri}. 
\end{proof}

\begin{lemma}
\label{tri2}
For all $t\geq 1$ and $n\geq 0$ we have
\begin{equation}
\label{ineq}
D\left(p^{(t)}|p^{(t+n)}\right)\leq D\left(p^{(t)}|\pi\right)-D\left(p^{(t+n)}|\pi\right).
\end{equation}
\end{lemma}
\begin{proof} Let us use induction on $n$.  The case $n=0$ is trivial.  Suppose (\ref{ineq}) has 
been verified for all $n'<n$.  When $n$ is even, we apply (\ref{ineq0}), the induction hypothesis, and Lemma \ref{tri} 
to obtain
\begin{align*}
D\left(p^{(t)}|p^{(t+n)}\right)&\leq D\left(p^{(t)}|p^{(t+n-1)}\right)\\
&\leq D\left(p^{(t)}|\pi\right)-D\left(p^{(t+n-1)}|\pi\right)\\
&\leq D\left(p^{(t)}|\pi\right)-D\left(p^{(t+n)}|\pi\right).
\end{align*}
When $n$ is odd, by (\ref{tri1}), the induction hypothesis, and then Lemma \ref{tri}, we have
\begin{align*}
D\left(p^{(t)}|p^{(t+n)}\right)=& D\left(p^{(t)}|p^{(t+1)}\right)+D\left(p^{(t+1)}|p^{(t+n)}\right)\\
\leq &D\left(p^{(t)}|p^{(t+1)}\right)+ D\left(p^{(t+1)}|\pi\right)\\
&-D\left(p^{(t+n)}|\pi\right)\\
= &D\left(p^{(t)}|\pi\right)-D\left(p^{(t+n)}|\pi\right). 
\end{align*}
\end{proof}

\begin{corollary} There exists some density $\pi^*$ such that 
$\lim_{t\rightarrow\infty} V\left(p^{(t)}, \pi^*\right)=0.$
\label{coro}
\end{corollary}

\begin{proof}  Pinsker's inequality and (\ref{ineq}) imply
$$\frac{1}{2}V^2\left(p^{(t)}, p^{(k)}\right)\leq \left|D\left(p^{(t)}|\pi\right)-D\left(p^{(k)}|\pi\right)\right|,$$
for all $t, k\geq 1.$  Because $D\left(p^{(t)}|\pi\right)$ is finite and decreases monotonically in $t$, 
$\lim_{t, k\rightarrow\infty} V\left(p^{(t)}, p^{(k)}\right)=0,$
i.e., $p^{(t)}$ is a Cauchy sequence in $L_1(\mathcal{X}\times \mathcal{Y})$.  Hence $p^{(t)}$ converges in 
$L_1(\mathcal{X}\times\mathcal{Y})$ to some density $\pi^*$.  (Only the completeness of $L_1(\mathcal{X}\times\mathcal{Y})$ is 
used here.  Further properties of $L_p$ spaces can be found in real analysis texts such as Royden \cite{R}.) \end{proof}

\begin{proposition}
In the setting of Corollary \ref{coro}, $\pi^*=\pi$.
\end{proposition}
\begin{proof}
Since $p^{(t)},\ t\geq 1,$ has the conditional $\pi_{X|Y}$ when $t$ is odd, and $\pi_{Y|X}$ when 
$t$ is even, the conditionals of $\pi^*$ must match those of $\pi$, i.e., 
\begin{equation}
\label{pistar}
\pi^*(x,y)=\pi^*_Y(y)\pi_{X|Y}(x|y)=\pi^*_X(x)\pi_{Y|X}(y|x),
\end{equation}
almost everywhere.  Under the assumption $\pi(x, y)>0$, (\ref{pistar}) implies
$$\pi^*_Y(y)=\pi^*_X(x)\frac{\pi_{Y|X}(y|x)}{\pi_{X|Y}(x|y)}=\pi^*_X(x)\frac{\pi_Y(y)}{\pi_X(x)}.$$
Integration over $y$ yields $1=\pi^*_X(x)/\pi_X(x)$, which, together with (\ref{pistar}), proves $\pi^*=\pi$. 
\end{proof}

Finally we finish the proof of Theorem \ref{main} by showing that the convergence in Corollary \ref{coro} also holds in 
relative entropy.
\begin{lemma}
$\lim_{t\rightarrow\infty} D\left(p^{(t)}|\pi\right)=0.$
\end{lemma}
\begin{proof}
We already have $D\left(p^{(t)}|\pi\right)\downarrow d$, say, with $d\geq 0$.  Taking $n\rightarrow\infty$ in (\ref{ineq}) we 
get
$$\liminf_{n\rightarrow\infty} D\left(p^{(t)}|p^{(t+n)}\right)\leq D\left(p^{(t)}|\pi\right)-d.$$
On the other hand, since 
$$D\left(p^{(t)}|p^{(t+n)}\right)=\int p^{(t)}\log \left( p^{(t)}/p^{(t+n)}\right) -p^{(t)}+p^{(t+n)}$$ 
and the integrand is non-negative, by Fatou's Lemma we have
\begin{equation}
\label{lsc}
\liminf_{n\rightarrow \infty} D \left(p^{(t)}|p^{(t+n)}\right)\geq D\left(p^{(t)}|\pi\right)
\end{equation}
which forces $d=0$.  The proof is now complete.  Note that (\ref{lsc}) is a case of the more general lower semi-continuity 
property of relative entropy (Csisz\'{a}r \cite{C75}). \end{proof}

\section{Remarks}  
As pointed out by an anonymous reviewer, the core of Section III consists of two parts: (i) showing
$\lim_{t\rightarrow\infty} V\left(p^{(t)}, \pi^*\right)=0$ for some $\pi^*$, whose conditionals match those of $\pi$, and (ii) 
showing that $\pi^*=\pi$.  Part (i) can be phrased more 
generally and is related to the results of Csisz\'{a}r and Shields (\cite{CS}, Theorem 5.1) on alternating I-projections.  It 
is also related to the information theoretic treatment of the EM algorithm (\cite{DLR} \cite{MV97}) of Csisz\'{a}r and Tusnady 
\cite{CT}.  The 
condition $\pi(x, y)>0$, not used in part (i), can be replaced by a weaker assumption, as long as one can show that there 
exists no density other than $\pi$ that possesses the two conditionals $\pi_{X|Y}$ and $\pi_{Y|X}$.

Lemma \ref{tri} appears in Yu \cite{Y}.  Lemmas \ref{prop} and \ref{tri2} are new.  Generalizations of Theorem \ref{main} to 
the Gibbs sampler with more than two components are possible (\cite{Y}), but technically more involved, because Lemmas 
\ref{prop} 
and \ref{tri2} are tailored to the two component case.  The issue of the 
rate of convergence, not addressed here, is definitely worth investigating. 

The DA algorithm has the following feature.  If we let $(X^{(0)}, Y^{(0)}, X^{(1)}, Y^{(1)}, \ldots)$ be the iterates 
generated, i.e., the conditional distribution of $Y^{(k)}|X^{(k)}$ is $\pi_{Y|X}$ and that of $X^{(k+1)}|Y^{(k)}$ is 
$\pi_{X|Y}$, then each of $\{X^{(k)}\}$ and $\{Y^{(k)}\}$ forms a reversible Markov chain.  Fritz \cite{F}, Barron \cite{B}, 
and Harremo\"{e}s and Holst \cite{HH} apply information theory to prove convergence theorems 
for reversible Markov chains.  Their results may be adapted to give an alternative (albeit less elementary) derivation of 
Theorem \ref{main}. 

\section*{Acknowledgments}
The author would like to thank Xiao-Li Meng and David van Dyk for introducing him to the topic of data augmentation.  He is 
also grateful to the Associate Editor and three anonymous reviewers for their valuable comments.


\end{document}